\documentclass[a4paper,11pt]{article}

\usepackage{amsfonts, amsmath, amssymb, amsthm}
\usepackage[T1]{fontenc} %
\usepackage[latin1]{inputenc} %
\usepackage[dvips]{graphicx}
\usepackage[normalsize]{subfigure}
\usepackage{color, fancybox, calc}
\usepackage{latexsym,amssymb,amsmath}
\usepackage{fullpage}
\usepackage{array,url, appendix}
\usepackage{refcount}
\usepackage{authblk}
\usepackage{enumitem}

\newtheorem{theorem}{Theorem}
\newtheorem{lemma}[theorem]{Lemma}

\newcommand{\C}{\mathcal{C}_{3, 2k\geq 6}}
\newcommand{\CCinq}{\mathcal{C}_{3, 5, 2k\geq 6}}

\newcommand{\defi}[1]{\emph{#1}}
\newcommand{\set}[1]{\{ #1 \}}

\title{Coloring graphs with no even hole $\geq 6$: the triangle-free case}
\author{Aur\'elie \textsc{Lagoutte}\thanks{The author is partially supported by the ANR Project \textsc{Stint} under \textsc{Contract ANR-13-BS02-0007}. Moreover most of this work was carried out while the author was affiliated to LIP, ENS de Lyon, France.}}

\affil{Univ. Grenoble Alpes, CNRS, Grenoble INP, G-SCOP, F-38000 Grenoble, France}

\date{Submitted: June 23, 2015 - Last update: \today}                                 %
\begin{document}

\maketitle

\begin{abstract} 
In this paper, we prove that the class of graphs with no triangle and no induced cycle of even length at least 6 has bounded chromatic number. It is well-known that even-hole-free graphs are $\chi$-bounded but we allow here the existence of $C_4$. The proof relies on the concept of Parity Changing Path, an adaptation of Trinity Changing Path which was recently introduced by Bonamy, Charbit and Thomass\'e to prove that graphs with no induced cycle of length divisible by three have bounded chromatic number.
\end{abstract}

\section{Introduction}

A \emph{hole} in a graph is an induced cycle of length at least four. A \emph{proper coloring} of a graph is a function that assigns to each vertex a color with the constraint that two adjacent vertices are not colored the same. The \emph{chromatic number} of $G$, denoted by $\chi(G)$, is the smallest number of colors needed to color the graph properly. All the colorings considered in the sequel are proper, so we just call them colorings. 
The size of the largest clique of $G$ is denoted $\omega(G)$.
We obviously have $\omega(G)\leq \chi(G)$, and one may wonder whether the equality holds. In fact, it does not hold in the general case, and the simplest counter-examples are \emph{odd holes}, i.e. holes of odd length, for which $\omega(G)=2$ but $\chi(G)=3$.  Graphs for which the equality $\chi(G')=\omega(G')$ holds for every induced subgraph $G'$ of $G$ are called \emph{perfect}, and the Strong Perfect Graph Theorem \cite{SPGT} proved that 
a graph if perfect if and only if it is
\emph{Berge}, that is to say there is no odd hole in $G$ nor in its complement.
In order to get some upper bound on $\chi(G)$, Gy\'arf\'as \cite{G54} introduced the concept of $\chi$-bounded class: a family $\mathcal{G}$ of graphs is called $\chi$\emph{-bounded}  if there exists a function $f$ such that $\chi(G')\leq f(\omega(G'))$  whenever $G'$ is an induced subgraph of $G\in \mathcal{G}$.

This notion has been widely studied since then, in particular in hereditary classes (\emph{hereditary} means closed under taking induced subgraph). 
A classical result of Erd\H{o}s \cite{E40} asserts that there exist graphs with arbitrarily large \emph{girth} (that is, the length of the shortest induced cycle) and arbitrarily large chromatic number. Thus forbidding only one induced subgraph $H$ may lead to a $\chi$-bounded class only if $H$ is acyclic. It is conjectured that this condition is also sufficient \cite{G32, S52}, but it is proved only if $H$ is a path, a star \cite{G54} or a tree of radius two \cite{TreeRadius2} (or three, with additional conditions \cite{TreeRadius3}). Scott \cite{ScottSubdivTree} also proved it for any tree $H$, provided that we forbid every induced subdivision of $H$, instead of just $H$ itself.

Consequently, forbidding holes in order to get a $\chi$-bounded class is conceivable only if we forbid infinitely many hole lengths. 
Two parameters should be taken into account: first, the length of the holes, and secondly, the parity of their lengths. 
In this respect, Gy\'arf\'as \cite{G54} made a famous series of three  conjectures. The first one asserts that the class of graphs with no odd hole is $\chi$-bounded. The second one asserts that, for every $k$, the class of graphs with no hole of length at least $k$ is $\chi$-bounded. The last one generalizes the first two conjectures and asserts that for every $k$, the class of graphs with no odd hole of length at least $k$ is $\chi$-bounded. After several partial results \cite{RS99, ScottCycles, CSSGyarfas}, the first and the second conjectures were recently solved by Chudnovsky, Scott and Seymour \cite{SSOddHoles, CSSLongHoles}. 
Moreover, we learned while writing this article that Scott and Seymour have proved a very general result implying the triangle-free case of the third conjecture (which also implies the result of this paper): for every $k\geq 0$, every triangle-free graph with large enough chromatic number admits a sequence of holes  of $k$ consecutive lengths \cite{SSConsecutiveHoles}\footnote{Bibliography update: while this article was under review process, Scott and Seymour finally managed to prove that, for every $c, k\geq 0$, every graph with clique number at most $c$ and sufficiently large chromatic number has a hole of every possible length modulo $k$ \cite{SSResidues}. This implies Gy\'arf\'as' third conjecture.}.

%
%

The class of even-hole-free graphs has  been extensively studied from a structural point of view. A decomposition theorem together with a recognition algorithm have been found by Conforti, Cornu\'ejols, Kapoor and Vu\v{s}kovi\'c \cite{EvenHole1, EvenHole2,ChangLu}. Reed conjectured \cite{R58} that every even-hole-free graphs has a vertex whose neighborhood is the union of two cliques (called a \emph{bisimplicial} vertex), which he and his co-authors proved \cite{EvenHoleBisimplicial} a few years later. As a consequence, they obtained that every even-hole-free graph $G$ satisfies $\chi(G)\leq 2\omega(G)-1$.

Forbidding $C_4$ is in fact a strong restriction since $C_4$ can also be seen as the complete bipartite graph $K_{2,2}$: K\"uhn and Osthus \cite{KuhnOsthus} proved that for every graph $H$ and for every integer $s$, every graph of large average degree (with respect to $H$ and $s$) with no $K_{s,s}$ as a (non-necessarily induced) subgraph contains an induced subdivision of $H$, where each edge is subdivided at least once. This strong result implies that the chromatic number is bounded in any class $\mathcal{C}$ defined as graphs with no triangles, no induced $C_4$ and no cycles of length divisible by $k$, for any fixed integer $k$. Indeed, let $G\in \mathcal{C}$ be a minimal counter-example to $\chi(G)\leq t$ (with $t$ chosen large enough with respect to $k$), then it has large minimum degree. Moreover it has neither induced $C_4$ nor triangles, consequently it has no $C_4$ subgraphs. By K\"uhn and Osthus' theorem, there exists an induced subdivision $H$ of $K_\ell$ for some well-chosen integer $\ell$ depending on $k$. Consider $K_\ell$ as an auxiliary graph where we color each edge with $c\in\{1,\ldots, k \}$ if this edge is subdivided $c$ times modulo $k$ in $H$. By Ramsey's theorem \cite{Ramsey}, if $\ell$ is large enough, then we can find a monochromatic clique $K$ of size $k$. Let $C_0$ be a Hamiltonian cycle through $K$ and call $C$ the corresponding cycle in the subdivided edges in $H$. Since $K$ was monochromatic in $K_\ell$, the edges used in $C_0$ are subdivided the same number of times modulo $k$, consequently $C$ has length divisible by $k$. Moreover, it is an induced cycle since each edge is subdivided at least once in $H$.

This is why we are interested in finding a $\chi$-boundedness result when every even hole except $C_4$ is forbidden, which was conjectured by  Reed \cite{ReedPrivate}. In this paper, we achieve a partial result by forbidding also triangles\footnote{The aforementioned recent result of Scott and Seymour \cite{SSResidues}, in addition to proving Gy\'arf\'as' third conjecture, also proves the general case of Reed's conjecture.}. This is a classical step towards $\chi$-boundedness, and Thomass\'e \emph{et al.} \cite{FPTBullFree} even asked whether this could always be sufficient, namely: does there exist a function $f$ such that for every class $\mathcal{C}$ of graphs and any $G\in \mathcal{C}$, $\chi(G)\leq f(\chi_T(G), \omega(G))$, where $\chi_T(G)$ denotes the maximum chromatic number of a triangle-free induced subgraph of $G$?

The result of this paper is closely related to the following recent one, by Bonamy, Charbit and Thomass\'e, answering to a question by Kalai and Meshulam on the sum of Betti numbers of the stable set complex (see \cite{Thomasse0mod3} for more details):

\begin{theorem}[\cite{Thomasse0mod3}]
There exists a constant $c$ such that every graph $G$ with no induced cycle of length divisible by 3 satisfies $\chi(G) < c$.
\end{theorem}

Indeed, the so-called Parity Changing Path (to be defined below) is directly inspired by their Trinity Changing Path. The structure of the proofs also have several similarities.

\bigskip

\paragraph{Contribution} We prove the following theorem:

\begin{theorem} \label{th: C chi borne}
There exists a constant $c$ such that every graph $G$ with no triangle and no induced cycle of even length at least 6  satisfies $\chi(G) < c$.
\end{theorem}

The outline is to prove the result when the 5-hole is also forbidden (see Lemma \ref{lem: sans C5} below), which should intuitively be easier, and then deduce the theorem for the general case.

To begin with, let us introduce and recall some notations:
the class under study, namely graphs with no triangle and no induced $C_{2k}$ with $k\geq 3$ (meaning that every even hole is forbidden except $C_4$) will be called \defi{$\C$} for short. Moreover, we will consider in Section \ref{sec: sans C5} the subclass \defi{$\CCinq$} of $\C$ in which the 5-hole is also forbidden. For two subsets of vertices $A, B \subseteq V$, $A$ \defi{dominates} $B$ if $B\subseteq N(A)$.
A \defi{major connected component} of $G$ is a connected component $C$ of $G$ for which $\chi(C)=\chi(G)$. Note that such a component always exists.  
For any induced path $P=x_1x_2\cdots x_\ell$ we say that $P$ is a path from its \defi{origin} $x_1$ to its \defi{end} $x_\ell$ or an \defi{$x_1x_\ell$-path}. Its \defi{interior} is $\{x_2, \ldots, x_{\ell-1}\}$ and its \defi{length} is $\ell-1$. 

Moreover, we use a rather common technique called a \emph{levelling} \cite{SSOddHoles, CSSGyarfas} :
given a vertex $v$, the \defi{$v$-levelling} is the partition $(N_0, N_1, \ldots, N_k, \ldots)$ of the vertices according to their distance to $v$: $N_k$ is the set of vertices at distance exactly $k$ from $v$ and is called the \defi{$k$-th level}. In particular, $N_0=\set{v}$ and $N_1=N(v)$. We need two more facts about levellings: 
if $x$ and $y$ are in the same part $N_k$ of a $v$-levelling, we call an \defi{upper $x y$-path} any shortest path from $x$ to $y$ among those with interior in $N_0\cup \cdots \cup N_{k-1}$. Observe that it always exists since there is an $xv$-path and a $vy$-path (but it may take shortcuts; in particular, it may be just one edge). Moreover, in any $v$-levelling, there exists $k$ such that $\chi(N_k)\geq \chi(G)/2$: 
indeed, if $t$ is the maximum of $\chi(N_i)$ over all levels $N_i$, one can color $G$ using $2t$ colors by coloring $G[N_{i}]$ with the set of colors $\{1, \ldots, t\}$ if $i$ is odd, and with the set of colors $\{t+1, \ldots , 2t\}$ if $i$ is even. Such a level with chromatic number at least $\chi(G)/2$ is called a \defi{colorful} level. Observe that, if $N_k$ is a colorful level in a triangle-free graph $G$ with $\chi(G)\geq 3$, then $k\geq 2$.

Let us now introduce the main tool of the proof, called \defi{Parity Changing Path} (\defi{PCP} for short) which, as already mentioned, is inspired by the \emph{Trinity Changing Path (TCP)} appearing in \cite{Thomasse0mod3}: intuitively (see Figure \ref{fig: PCP} for an unformal diagram), a PCP is a sequence of induced subgraphs and paths $(G_1, P_1, \ldots, G_\ell, P_\ell, H)$ with no "bad" chord between them, such that each block $G_i$ can be crossed by two possible paths of different parities, and the last block $H$ typically is a "stock" of big chromatic number, in which we can find whichever structure always appears in a graph with high chromatic number.
Formally, a \defi{PCP of order $\ell$} in $G$ is a sequence of induced subgraphs $G_1, \ldots, G_{\ell}, H$ (called \defi{blocks}; the $G_i$ are the \defi{regular} blocks) and induced paths $P_1, \ldots , P_{\ell}$ such that the origin of $P_i$ is some vertex $y_i$ in $G_i$, and the end of $P_i$ is some vertex $x_{i+1}$ of $G_{i+1}$ (or of $H$ if $i=\ell$). Apart from these special vertices which belong to exactly two subgraphs of the PCP, the blocks and paths $G_1, \ldots, G_\ell, H, P_1, \ldots, P_\ell$ composing the PCP are pairwise disjoint. The only possible edges have both endpoints belonging to the same block or path.
We also have one extra vertex $x_1\in G_1$ called the \defi{origin} of the PCP. 
Moreover in each block $G_i$, there exists one induced $x_i y_i$-path of odd length, and one induced $x_i y_i$-path of even length (these paths are not required to be disjoint one from each other). In particular $x_i \neq y_i$ and $x_i y_i$ is not an edge. For technical reasons that will appear later, we also require that $H$ is connected, every $G_i$ has chromatic number at most 4 and every $P_i$ has length at least 2.  Finally the chromatic number of $H$ is called the \defi{leftovers}. 

In fact in Section \ref{sec: Yes C4 avec C5}, we need a slightly stronger definition of PCP: a \defi{strong PCP} is a PCP for which every $G_i$ contains an induced $C_5$.

\begin{figure}
\center
\includegraphics[scale=1]{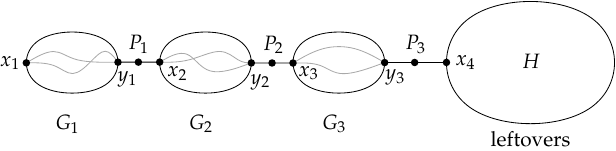}
\caption{An informal diagram for a PCP of order 3. Grey curved lines stand for the even and odd length $x_iy_i$-paths.}
\label{fig: PCP}
\end{figure}

\bigskip

We first bound the chromatic number in $\CCinq$ (see Lemma \ref{lem: sans C5} below), which is easier because we forbid one more cycle length, and then deduce the theorem for $\C$. The proofs for $\C$ and $\CCinq$ follow the same outline, which we informally describe here:

\begin{enumerate}[label=(\roman*)]
\item \label{step1} If $\chi(G)$ is large enough, then for every vertex $v$ we can grow a PCP whose origin is $v$ and whose leftovers are large (Lemmas \ref{lem: debut PCP sans C5}, \ref{lem: existence PCP sans C5} and then Lemma \ref{lem: debut PCP}).
\item \label{step2} Using \ref{step1}, if $\chi(G)$ is large enough and $(N_0, N_1, \ldots)$ is the $v$-levelling,
we can grow a \defi{rooted} PCP:
it is a PCP in a level $N_k$, which has a \defi{root}, \emph{i.e.} a vertex in the previous level $N_{k-1}$ whose unique neighbor in the PCP is the origin (Lemma \ref{lem: rooted PCP sans C5} and then Lemma \ref{lem: rooted PCP}).
\item \label{step3} Given a rooted PCP in a level $N_k$, if a vertex $x\in N_{k-1}$ has a neighbor in some block, then it has a neighbor in every preceding regular block (Lemma \ref{obs: sommets riches sans C5}).
\item \label{step4} Given a rooted PCP of order $\ell$ in a level $N_k$ and a stable set $S$ in $N_{k-1}$, the chromatic number of $N(S)\cap N_k$ is bounded. Consequently, the \defi{active lift} of the PCP, defined as $N(G_\ell)\cap N_{k-1}$, has high chromatic number (Lemmas \ref{lem: shadow stable bornee sans C5}, \ref{lem: active lift gros si neighb stable gros} and \ref{lem: shadow borne bornee sans C5} and then Lemmas \ref{lem: active lift gros si neighb stable gros}, \ref{lem: shadow stable bornee}, \ref{lem: shadow borne bornee}).
\item \label{step5} The final proofs put everything together: consider a graph of $G\in \CCinq$ (resp. $\C$) with chromatic number large enough. Then pick a vertex $v$, let $(N_0, N_1, \ldots)$ be the $v$-levelling and $N_k$ be a colorful level. By \ref{step2}, grow inside $N_k$ a rooted PCP $P$. Then by \ref{step4}, get an active lift $A$ of $P$ inside $N_{k-1}$ with big chromatic number. Grow a rooted PCP $P'$ inside $A$, and get an active lift $A'$ of $P'$ inside $N_{k-2}$ with chromatic number big enough to find an edge $xy$ (resp. a 5-hole $C$) in $A'$ . Then ``clean" $P'$ in order to get a stable set $S$ inside the last regular block of $P'$, dominating this edge (resp. hole). Now find an even hole of length $\geq 6$ in $\set{x,y}\cup S\cup P$ (resp. $C\cup S\cup P$), a contradiction.
\end{enumerate}

\section{Forbidding 5-holes}
\label{sec: sans C5}

This section is devoted to the proof of the following lemma :

\begin{lemma} \label{lem: sans C5}
There exists a constant $c'$ such that every graph $G\in \CCinq$ satisfies $\chi(G)< c'$.
\end{lemma}

We follow the outline described above. Let us start with step \ref{step1}:

\begin{lemma} \label{lem: debut PCP sans C5}
Let $G\in \CCinq$ be a connected graph and $v$ be any vertex of $G$. For every $\delta$ such that $\chi(G)\geq \delta \geq 18$, there exists a PCP of order 1 with origin $v$ and leftovers at least $h(\delta)=\delta/2 -8$.
\end{lemma}

\begin{proof}
The proof is illustrated on Figure \ref{fig: build PCP overview}. Let $(N_0, N_1, \ldots)$ be the $v$-levelling and $N_k$ be a colorful level (hence $k\geq 2$ since $G$ is triangle-free).
Let $N'_k$ be a major connected component of $G[N_k]$, so $\chi(N'_k)\geq \delta/2$.
 Let $xy$ be an edge of $N'_k$, and $x'$ (resp. $y'$) be a neighbor of $x$ (resp. $y$) in $N_{k-1}$.
Let  $Z' = N(\{x',y',x,y\})\cap N'_k$ and $Z=Z' \setminus \{x,y\}$. 
 Let $z\in Z$ be a vertex having a neighbor $z_1$ in a major connected component $M_1$ of $N'_k \setminus Z'$. Observe that $N'_k \setminus Z'$ is not empty since  $\chi(Z')\leq 6$ (the neighborhood of any vertex is a stable set since $G$ is triangle-free).
 The goal is now to find two $vz$-paths $P$ and $P'$ of different parities with interior in $G[N_0 \cup \ldots \cup \{ x', y' \} \cup \{x,y\}]$. Then we can set $G_1=G[P\cup P']$, $P_1=G[\{z,z_1\}]$ and $H=G[M_1]$ as parts of the wanted PCP. In practice, we need to be a little more careful to ensure the condition on the length of $P_1$ and the non-adjacency between $z$ and $H$, which is described after finding such a $P$ and a $P'$.

Let $P_0$ (resp. $P'_0$) be a $vx'$-path (resp. $vy'$-path) of length $k-1$ (with exactly one vertex in each level).
By definition of $Z$, $z$ is connected to $\{x',y',x,y\}$.
\begin{enumerate}
\item \label{case1 lem: debut PCP sans C5}(see Figure \ref{fig: build PCP case 1}) If $z$ is connected to $x$ or $y$, say $x$, then $z$ is connected neither to $x'$ nor to $y$, otherwise it creates a triangle. We add the path $x'xz$ to $P_0$ to form $P$. Similarly, we add either the edge $y'z$ if it exists, or else the path $y' y x z$ to $P'_0$ to form $P'$. Observe that $P'$ is indeed an induced path since there is no triangle.
 Moreover, the lenghts of $P$ and $P'$ differ by exactly one, so $P$ and $P'$ have indeed distinct parities.
\item \label{case2 lem: debut PCP sans C5} (see Figure \ref{fig: build PCP case 2}) Otherwise, $z$ is connected neither to $x$ nor to $y$, thus $z$ is connected to exactly one of $x'$ and $y'$, since otherwise it would either create a triangle $x', y', z$ or a 5-hole $z x' x y y'$, so say $zx'\in E$ and $zy'\notin E$. We add the edge $x'z$ to $P_0$ to form $P$. We add the path $y' x' z$ if $y'x'\in E$, otherwise add the path $y' y x x' z$ to $P'_0$ to form $P'$. Observe that this is an induced path since $G$ has no triangle and no 5-hole.  Moreover, the lenghts of $P$ and $P'$ differ by either one or three, so $P$ and $P'$ have indeed distinct parities.
\end{enumerate}
Now comes the fine tuning.
Choose in fact $z_1\in M_1\cap N(z)$ so that $z_1$ is connected to a major connected component $M_2$ of $M_1 \setminus N(z)$. Choose $z_2$ a neighbor of $z_1$ in $M_2$ such that $z_2$ is connected to a major connected component $M_3$ of $M_2\setminus N(z_1)$. We redefine $H=G[\{z_2 \cup M_3\}]$ and $P_1=G[\{z,z_1,z_2\}]$. Then $P_1$ is a path of length 2, $G_1$ is colorable with 4 colors as the union of two induced paths, and $H$ is connected. Moreover $H$ has chromatic number at least $\chi(N'_k)-\chi(Z')-\chi(N(z))-\chi(N(z_1))$. Since the neighborhood of any vertex is a stable set, $\chi(Z')\leq 6$ and $\chi(N(z)), \chi(N(z_1))\leq 1$. Thus $\chi(H)\geq \delta/2- 8$.

\end{proof}

\begin{figure}
\center

\subfigure[Overview of the situation\label{fig: build PCP overview}]{\includegraphics[scale=1, page=1]{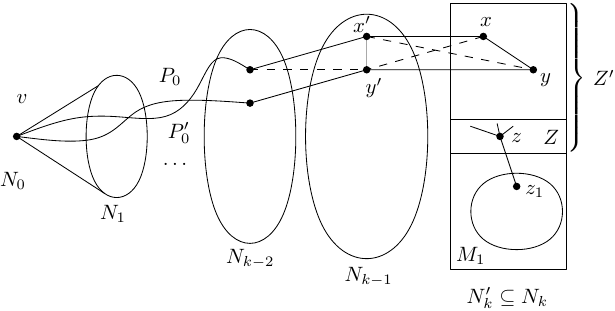}}
\subfigure[Case \ref{case1 lem: debut PCP sans C5}\label{fig: build PCP case 1}]{\includegraphics[scale=1, page=2]{fig/buildPCP}}

\subfigure[Case \ref{case2 lem: debut PCP sans C5}\label{fig: build PCP case 2}]{\includegraphics[scale=1, page=3]{fig/buildPCP}}
\caption{Illustrations for the proof of Lemma \ref{lem: debut PCP sans C5}. Dashed edges stand for non-edges, and grey edges stand for edges that may or may not exist.}
\label{fig: build PCP}
\end{figure}

We can iterate the previous process to grow some longer PCP. In the following, for a function $f$ and an integer $k$, $f^{(k)}$ denotes the $k$-th iterate of $f$, that is to say that $f^{(k)}(x)=\underbrace{(f \circ \ldots \circ f)}_{k \text{ times}} (x)$.

\begin{lemma}
\label{lem: existence PCP sans C5}
Let $h(x)=x/2-8$ be the function defined in Lemma \ref{lem: debut PCP sans C5}. For every positive integers $\ell, \delta\in \mathbb{Z}_+$, if $G\in \CCinq$ is connected and satisfies $\chi(G)\geq \delta$ and $h^{(\ell-1)}(\delta)\geq 18$, then from any vertex $x_1$ of $G$, one can grow a PCP of order $\ell$ with leftovers at least $h^{(\ell)}(\delta)$.
\end{lemma}

\begin{proof}
We prove the result by induction on $\ell$. For $\ell=1$, the result follows directly from Lemma \ref{lem: debut PCP sans C5}. Now suppose it is true for $\ell-1$, and let $G\in \CCinq$ be such that $\chi(G)\geq \delta$ and $h^{(\ell-1)}(\delta)\geq 18$. Then $\delta\geq h^{(\ell-1)}(\delta)\geq 18$, so we can apply Lemma \ref{lem: debut PCP sans C5} to get a PCP of order 1 and leftovers at least $h(\delta)$ from any vertex $x_1$. Let $x_2$ be the unique vertex in the intersection of $P_1$ and the last block $H$ of the PCP (as in the definition). Now apply the induction hypothesis to $H$, knowing that $H$ is connected, $\chi(H)\geq h(\delta)=\delta'$ and $h^{(\ell-2)}(\delta')\geq 18$. Then we obtain a PCP of order $\ell-1$ with origin $x_2$ and leftovers at least $h^{(\ell-2)}(\delta')$, which finishes the proof by gluing the two PCP together.
\end{proof}

Now we grow the PCP in a level $N_k$ of high chromatic number, and we want the PCP to be rooted (\emph{i.e.} there 
exists a root $u'\in N_{k-1}$ that is adjacent to the origin $u$ of the PCP, but to no other vertex of the PCP). This is step \ref{step2}.

\begin{lemma}
\label{lem: rooted PCP sans C5}
Let $G\in \CCinq$ be a connected graph, $v\in V(G)$ and $(N_0, N_1, \ldots)$ be the $v$-levelling. Let $h$ be the function defined in Lemma \ref{lem: debut PCP sans C5}. For every $k, \delta$ such that $\chi(N_k)\geq \delta+1$ and $h^{(\ell-1)}(\delta)\geq 18$, there exists a rooted PCP of order $\ell$ in $N_k$ with leftovers at least $h^{(\ell)}(\delta)$.
\end{lemma}

\begin{proof}
Let $N'_k$ be a major connected component of $N_k$ and let $u'$ be a vertex in $N_{k-1}\cap N(N'_k)$ (which is non-empty). Since there is no triangle, $N'_k\setminus N(u')$ still has big chromatic number (at least $\delta$), and let $N''_k$ be a major connected component of $N'_k\setminus N(u')$. Let $z$ be a vertex of $N(u')\cap N'_k$ having a neighbor in $N''_k$. Then we apply Lemma \ref{lem: existence PCP sans C5} in $\{z\}\cup N''_k$ to grow a PCP of order $\ell$ from $z$ with leftovers at least $h^{(\ell)}(\delta)$. Now $u'$ has an only neighbor $z$ on the PCP, which is the origin.
\end{proof} 
 
Let us observe the properties of such a rooted PCP. We start with step \ref{step3}:

\begin{lemma} \label{obs: sommets riches sans C5}
Let $v$ be a vertex of a graph $G\in \C$, $(N_0, N_1, \ldots)$ be the $v$-levelling. Let $(G_1, P_1, \ldots, G_\ell, P_\ell, H)$ be a rooted PCP of order $\ell$ in a level $N_k$ for some $k$. If $x'\in N_{k-1}$ has a neighbor $x$ in some regular block $G_{i_0}$ (resp. in $H$), then $x$ has a neighbor in every $G_i$ for $1\leq i \leq i_0$ (resp. for $1\leq i \leq \ell$).
\end{lemma}

\begin{proof} If $x'$ has a neighbor in $H$, we set $i_0=\ell+1$.
We proceed by contradiction. Let $u$ be the origin of the PCP and $u'$ its root. Since $x'\neq u'$ by definition of the root, there exists an upper $x'u'$-path $P_{up}$ of length at least one. 
Consider a $ux$-path inside the PCP. Let $v_1, \ldots, v_r$ be the neighbors of $x'$ on this path, different from $x$ (if any), in this order (from $u$ to $x$).
Now we can show that any regular block $G_i$ with $1\leq i \leq i_0-1$ contains at least one $v_j$: suppose not for some index $i$, let $j$ be the greatest index such that $v_j$ is \emph{before} $x_i$, 
\emph{i.e.} $v_j\in G_1\cup P_1 \cup \cdots \cup G_{i-1} \cup P_{i-1}$.

If such an index does not exist (\emph{i.e.} all the $v_j$ are after $G_i$), then there is an odd and an even path from $u$ to $v_1$ of length at least 3 by definition of a regular block,
and this path does not contain any neighbor of $x'$. Close them to build two induced cycles by going through 
$x'$, $P_{up}$ and $u'$: one of them is an even cycle, and its length is at least 6.

If $j=r$ (\emph{i.e.} all the $v_j$ are before $G_i$), then we can use the same argument with a path of well-chosen parity from $v_r$ to $x$, crossing $G_i$.

Otherwise, there is an odd and an even path in the PCP between $v_j$ and $v_{j+1}$, crossing $G_i$, and its length is at least 4 because $x_i$ and $y_i$ are at distance at least 2 one from each other. We can close the even path by going back and forth to $x$: this gives an even hole of length at least 6.
\end{proof}

Note that, in the lemma above, $G$ is taken in $\C$ and not in $\CCinq$. In particular, we will use Lemma \ref{obs: sommets riches sans C5} in the next section as well. Let us now continue with step \ref{step4}:
 
\begin{lemma} \label{lem: shadow stable bornee sans C5}
Let $v$ be a vertex of a graph $G\in \CCinq$ and $(N_0, N_1, \ldots)$ be the $v$-levelling. 
Let $S\subseteq N_{k-1}$ be a stable set. Then $\chi(N(S)\cap N_k)\leq 52$.

\end{lemma}

\begin{proof}
Let $\delta=\chi(N(S)\cap N_{k-1})-1$. Suppose by contradiction that $\delta\geq 52 $, then $h(\delta)\geq 18$ hence by Lemma \ref{lem: rooted PCP sans C5}, we can grow a rooted PCP of order 2 inside $N(S)\cap N_k$. Let $u$ be the origin of the PCP and $u'$ its root. Observe in particular that $S$ dominates $G_2$. Let $xy$ be an edge of $G_2$, and let $x'$ (resp. $y'$) be a neighbor of $x$ (resp. $y$) in $S$. 
By Lemma \ref{obs: sommets riches sans C5}, both $x'$ and $y'$ have a neighbor in $G_1$.
This gives an $x'y'$-path $P_{down}$ with interior in $G_1$. In order not to create an even hole nor a 5-hole by closing it with $x' x y y'$, we can ensure that $P_{down}$ is an even path of length at least 4. Moreover, there exists an upper $x'y'$-path $P_{up}$. Then either the hole formed by the concatenation of $P_{up}$ and $x' x y y'$, or the one formed by the concatenation of $P_{up}$ and $P_{down}$ is an even hole of length $\geq 6$, a contradiction.
\end{proof}

The previous lemma allows us to prove that one can \emph{lift} the PCP up into $N_{k-1}$ to get a subset of vertices with high chromatic number. We state a lemma that will be reused in the next section:

%

\begin{lemma}
\label{lem: active lift gros si neighb stable gros}
Let $v$ be a vertex of a graph $G\in \C$ and $(N_0, N_1, \ldots)$ be the $v$-levelling. Let $P$ be a rooted PCP of order $\ell\geq 1$ in a level $N_k$  with leftovers at least $\delta$ (hence $k\geq 2$). Let $A=N(G_\ell)\cap N_{k-1}$ (called the active lift of the PCP).
Suppose that for every stable set $S\subseteq A$, we have $\chi(N(S)\cap N_{k})\leq \gamma$,
 then $\chi(A)\geq {\delta}/{\gamma}$.
\end{lemma}

\begin{proof}
Let $r=\chi(A)$, suppose by contradiction that $r<\delta/\gamma$ and decompose $A$ into $r$ stable sets $S_1, \ldots, S_r$. Then $N(A)\cap N_k$ is the (non-necessarily disjoint) union of $r$ sets $N(S_1)\cap N_k, \ldots, N(S_r)\cap N_k$, and each of them has chromatic number at most $\gamma$ by assumption. Consequently $\chi(N(A)\cap N_k)\leq r \gamma < \delta$ and hence $\chi(H\setminus N(A)) \geq \chi(H) - \chi(N(A)\cap N_k) \geq 1$.
Let $x$ be any vertex of $H\setminus N(A)$ and $x'$ be a neighbor of $x$ in $N_{k-1}$. By construction, $x'\notin A$ so $x'$ has no neighbor in $G_\ell$. This is a contradiction with Lemma \ref{obs: sommets riches sans C5}.
\end{proof}

By Lemmas \ref{lem: shadow stable bornee sans C5} and \ref{lem: active lift gros si neighb stable gros} with $\gamma=52$, we can directly deduce the following:

\begin{lemma} \label{lem: shadow borne bornee sans C5}
Let $v$ be a vertex of a graph $G\in \CCinq$ and $(N_0, N_1, \ldots)$ be the $v$-levelling. Let $P$ be a rooted PCP of order $\ell\geq 1$ in a level $N_k$ with leftovers at least $\delta$ (hence $k\geq 2$). Let $A=N(G_\ell)\cap N_{k-1}$ be the active lift of the PCP, then we have $\chi(A)\geq g(\delta)={\delta}/{52}$.
\end{lemma}

%

We can now finish the proof, this is step \ref{step5}. Recall that a sketch was provided, and it may help to understand the following proof. Moreover, Figure \ref{fig: PCP conclude} illustrates the proof.

\begin{figure}
\center
\def\echelle{0.8025}

\subfigure[First paragraph of the proof\label{fig: PCP conclude sans C5 milieu}]{\includegraphics[scale=\echelle, page=5]{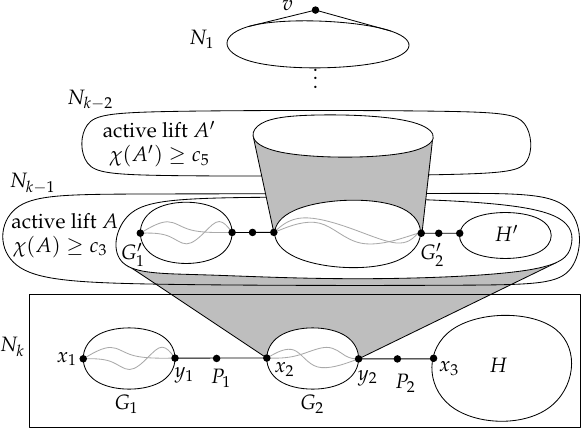}}
\subfigure[End of the proof\label{fig: PCP conclude sans C5 fin}]{\includegraphics[scale=\echelle, page=11]{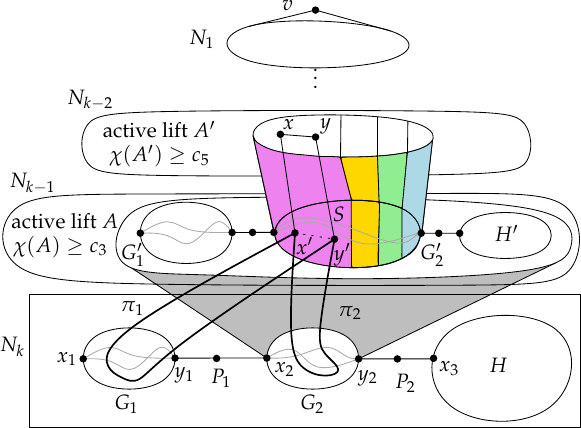}}

\caption{Illustrations for the proof of Lemma \ref{lem: sans C5}.}
\label{fig: PCP conclude}

\end{figure}

\begin{proof}[Proof of Lemma \ref{lem: sans C5}] 
%
Let $c'$ be a constant big enough so that 
$$g\left(h^{(2)}\left(g\left(h^{(2)}\left(\frac{c'}{2}-1\right)\right)-1\right)\right)\geq 5 \ .$$
Suppose that $\chi(G)\geq c'$. Pick a vertex $v$, let $(N_0, N_1, \ldots)$ be the $v$-levelling and $N_k$ be a colorful level, so $\chi(N_k)\geq \chi(G)/2\geq c_1+1$ where $c_1=c'/2-1$. By Lemma \ref{lem: rooted PCP sans C5}, grow a rooted PCP  $P=(G_1, P_1, G_2, P_2, H)$ inside $N_k$ of order 2 with leftovers at least $c_2=h^{(2)}(c_1)$. Then apply Lemma \ref{lem: shadow borne bornee sans C5} and get an active lift $A$ of $P$ inside $N_{k-1}$ with chromatic number at least $c_3=g(c_2)$. By definition of $c'$ and $c_3$, we can show that $c_3\geq 2$, which in particular implies that $k-1\geq 2$, since $N_1$ is a stable set. Since $h(c_3 -1)\geq 18$, apply again Lemma \ref{lem: rooted PCP sans C5} to get a rooted PCP $P'=(G'_1, P'_1, G'_2, P'_2, H')$ of order 2 inside $N_{k-1}$ with leftovers at least $c_4=h^{(2)}(c_3 -1)$. Now apply Lemma \ref{lem: shadow borne bornee sans C5} to get an active lift $A'$ of $P'$ inside $N_{k-2}$ with chromatic number at least $c_5=g(c_4)$. The situation at this point is described in Figure \ref{fig: PCP conclude sans C5 milieu}.

Because of the chromatic restriction in the definition of the PCP, one can color $G'_2$ with 4 colors. Moreover, $G'_2$ dominates $A'$ by definition. Thus there exists a stable set $S\subseteq G'_2$ such that $\chi(N(S)\cap A')\geq c_6=c_5/4$ (since $A'$ is the union of the $N(S')\cap A'$ for the four stable sets $S'$ that partition $G_2'$).

Now $c_6>1$ so there is an edge $xy$ inside $N(S) \cap A'$. Call $x'$ (resp. $y'$) a vertex of $S$ dominating $x$ (resp. $y$). Both $x'$ and $y'$ have a neighbor in $G_2$ by definition of $A$ and, by Lemma \ref{obs: sommets riches sans C5}, both $x'$ and $y'$ also have a neighbor in $G_1$.  This gives an $x'y'$-path $\pi_1$ (resp. $\pi_2$) with interior in $G_1$ (resp. $G_2$). Due to the path $x'x y y'$ of length 3, $\pi_1$ and $\pi_2$ must be even paths of length at least 4 (see Figure \ref{fig: PCP conclude sans C5 fin}). Thus the concatenation of $\pi_1$ and $\pi_2$ is an even hole of length at least 6, a contradiction.
\end{proof}

\section{General case}
\label{sec: Yes C4 avec C5}

This section aims at proving Theorem \ref{th: C chi borne}, using the result of the previous section. As already mentioned, we follow the same outline, except that we now need the existence of a $C_5$ several times. Let us start by a technical lemma to find both an even and an odd path out of a 5-hole and its dominating set:

\begin{lemma}
\label{lem: C5 chemin pair chemin impair}
Let $G$ be a triangle-free graph inducing a 5-hole $C$. Let $S\subseteq V(G)$ be a minimal dominating set of $C$, assumed to be disjoint from $C$. If we delete the edges with both endpoints in $S$, then for every vertex $t\in S$, there exists a vertex $t'\in S$ such that 
one can find an induced $tt'$-path of length 4 and an induced $tt'$-path of length 3 or 5, both with interior in $C$.
\end{lemma}

\begin{proof}
Let $t\in S$. Since $S$ is a minimal dominating set of $C$, there exists $v_1\in C$ such that $N(v_1)\cap S=\{t\}$ (otherwise, $t$ is useless and can be removed from $S$). Number the other vertices of $C$ with $v_2, \ldots , v_5$ (following the adjacency on the cycle). Since $G$ is triangle-free, $t$ can not be adjacent to both $v_3$ and $v_4$, so up to relabeling the cycle in the other direction we assume that $t$ is not adjacent to $v_3$. Let $t'\in S$ be a vertex dominating $v_3$. Then $tv_1v_2v_3t'$ is an induced path of length 4 between $t$ and $t'$. Moreover, $tv_1v_5v_4v_3t'$ is a (non-necessarily induced) path of length 5 between $t$ and $t'$. If this path is not induced, the only possible chords are $tv_4$ and $t'v_5$ since $G$ is triangle-free, which in any case gives an induced $tt'$-path of length 3.
\end{proof}

Recall that in this section, we are interesting in strong PCP, \emph{i.e.} PCP, all regular blocks $G_i$ of which contain an induced $C_5$. We start with step \ref{step1}:

\begin{lemma} \label{lem: debut PCP}
Let $c'$ be the constant of Lemma \ref{lem: sans C5}, let $G\in \C$ and $v$ be any vertex of $G$. For every $\delta \in \mathbb{N}$ such that $\chi(G)\geq \delta \geq 2c'$, there exists a strong PCP of order 1 with origin $v$ and leftovers at least $f(\delta)=\delta/2-15$.

\end{lemma}

\begin{proof}

Let $(N_0, N_1, \ldots)$ be the $v$-levelling, $N_k$ be a colorful level (hence $k\geq 2$ since $G$ is triangle-free) and let $N'_k$ be a major connected component of $G[N_k]$, so $\chi(N'_k)\geq c'$. Using Lemma \ref{lem: sans C5}, there exists a 5-hole $C$ in $G[N'_k]$. Consider a minimum dominating set $D$ of $C$ inside $N_{k-1}$.

From now on, the proof is very similar to the one of Lemma \ref{lem: debut PCP sans C5}. Similarly, we define $Z'=N(D\cup C)\cap N'_k$ and $Z=Z' \setminus C$. Let $z\in Z$ be a vertex having a neighbor $z_1$ in a major connected component $M_1$ in $N'_k \setminus Z'$. The goal is now to find two $vz$-paths $P$ and $P'$ of different parity with interior in $N_0 \cup \ldots \cup D \cup C$, then we can set $G_1=G[P\cup P'\cup C]$, $P_1=G[\{z,z_1\}]$ and $H=G[M_1]$ as parts of the wanted PCP. In practice, we need to be a little more careful to ensure the condition on the length of $P_1$ and the non-adjacency between $z$ and $H$.

Let us now find those two paths $P$ and $P'$. By definition of $Z$, $z$ also has a neighbor in $D$ or in $C$. 
\begin{enumerate}
\item If $z$ has a neighbor $x\in C$, let $y\in C$ be a vertex adjacent to $x$ on the hole. Let $x'$ and $y'$ be respectively a neighbor of $x$ and a neighbor of $y$ in $D$, in particular $x'\neq y'$ since $G$ is triangle-free. Observe that $z$ is connected neither to $x'$ nor to $y$, otherwise it creates a triangle. We grow $P$ by starting from an induced path of length $k-1$ from $v$ to $x'$ and then add the path $x'x z$. Similarly, we grow $P'$ by starting from an induced path of length $k-1$ from $v$ to $y'$, and then add the edge $y'z$ if it exists, or else the path $y' y x z$. Observe that $P'$ is indeed an induced path since there is no triangle. Moreover $P$ and $P'$ have distinct parities.
\item If $z$ has no neighbor in $C$, then it has at least one neighbor $x'$ in $D$. Apply Lemma \ref{lem: C5 chemin pair chemin impair} to get a vertex $y'\in D$ such that there exists an $x'y'$-path of length 3 or 5, and another one of length 4, both with interior in $C$. Observe that 
$x'$ and $y'$ cannot have a common neighbor $u$ in $N_{k-2}\cup \{z\}$, otherwise there would be either a triangle $x', u, y'$ (if $x'y'\in E$), or a $C_6$ using the \mbox{$x'y'$-path} of length 4 with interior in $C$. Now we grow $P$ by starting from an induced path of length $k-1$ from $v$ to $x'$, and add the edge $x'z$. We grow $P'$ by starting from an induced path of length $k-1$ from $v$ to $y'$, and then add the edge $x'y'$ if it exists, otherwise add the $x'y'$-path of length 3 or 5 with interior in $C$, and then finish with the edge $x'z$.
\end{enumerate}

Now comes the fine tuning.
Choose in fact $z_1\in M_1\cap N(z)$ so that $z_1$ is connected to a major connected component $M_2$ of $M_1 \setminus N(z)$. Choose $z_2$ a neighbor of $z_1$ in $M_2$ such that $z_2$ is connected to a major connected component $M_3$ of $M_2\setminus N(z_1)$. We redefine $H=G[\{z_2 \cup M_3\}]$ and $P_1=G[\{z,z_1,z_2\}]$. Then $P_1$ is a path of length 2, $H$ is connected and $G_1$ is colorable with 4 colors (it is easily 7-colorable as the union of a 5-hole and two paths; a careful case analysis shows that it is 4-colorable). Moreover $H$ has chromatic number at least $\chi(N'_k)-\chi(Z')-\chi(N(z))-\chi(N(z_1))$. Since the neighborhood of any vertex is a stable set, $\chi(Z')\leq |D|+|C| +\chi(C)\leq 13$ and $\chi(N(z)), \chi(N(z_1))\leq 1$. Thus $\chi(H)\geq \delta/2- 15$.

\end{proof}

We go on with step \ref{step2}: find a strong rooted PCP.
The following lemma is proved in the same way as Lemma \ref{lem: rooted PCP sans C5} by replacing the use of Lemma \ref{lem: existence PCP sans C5} by Lemma \ref{lem: debut PCP}, so we omit the proof here.

\begin{lemma}
\label{lem: rooted PCP}
Let $G\in \C$ be a connected graph, $f$ be the function defined in Lemma \ref{lem: debut PCP}, $v$ be a vertex of $G$ and $(N_0, N_1, \ldots)$ be the $v$-levelling. For every $k, \delta$ such that $\chi(N_k)\geq \delta+1\geq 2c'+1$, there exists a strong rooted PCP of order 1 in $N_k$ with leftovers at least $f(\delta)$.
\end{lemma}

Step \ref{step3} is proved by Lemma \ref{obs: sommets riches sans C5} from the previous section, and was valid not only for $G\in \CCinq$ but also for $G\in \C$. So we continue with step \ref{step4}:

\begin{lemma} \label{lem: shadow stable bornee}
Let $v$ be a vertex of a graph $G\in \C$, $(N_0, N_1, \ldots)$ be the $v$-levelling. 
Let $S$ be a stable set inside $N_{k-1}$. Then $\chi(N(S)\cap N_k)\leq 2c'$.
\end{lemma}

\begin{proof}
Suppose by contradiction that $\chi(N(S)\cap N_k)\geq 2c'+1$. By Lemma \ref{lem: rooted PCP}, we can grow in $N(S)\cap N_k$ a rooted PCP of order 1, and in particular $S$ dominates $G_1$.
By definition of a strong PCP, there is a 5-hole $C$ in $G_1$. Since $S$ is a dominating set of $C$, we can apply Lemma \ref{lem: C5 chemin pair chemin impair} to get two vertices $t, t'\in S$ such that one can find both an even and an odd $tt'$-path with interior in $C$ and length at least 3. Then any upper $tt'$-path close a hole of even length $\geq 6$.
\end{proof}

In fact, as in previous section, we can directly deduce from Lemmas \ref{lem: active lift gros si neighb stable gros} and \ref{lem: shadow stable bornee} that one can lift the PCP up into $N_{k-1}$ to get a subset of vertices with high chromatic number:

\begin{lemma} \label{lem: shadow borne bornee}
Let $G\in \C$, $v\in V(G)$ and $(N_0, N_1, \ldots)$ be the $v$-levelling. Let $P$ be a strong rooted PCP of order $1$ in a level $N_k$ (hence $k\geq 2$) with leftovers $\delta$. Let $A=N(G_1)\cap N_{k-1}$ be the active lift of the PCP. If $\delta \geq 2c'$, then $\chi(A)\geq \varphi(\delta)=\frac{\delta}{2c'}$.
\end{lemma}

We are now ready to finish the proof, this is step \ref{step5}. 
The proof follow the same outline as the proof of Lemma \ref{lem: sans C5}, which was sketched at the end of the Introduction.

Recall that a sketch was given and may be useful to have a less technical overview of the proof.

\begin{proof}[Proof of Theorem \ref{th: C chi borne}] 
%
Let $c$ be a constant such that 
$$\varphi\left(f\left(\varphi\left(f\left(\frac{c}{2} -1\right)\right)-1\right)\right)\geq 4c' \ .$$ 
Suppose that $G\in \C$ has chromatic number $\chi(G)\geq c$. Then pick a vertex $v$, let $(N_0, N_1, \ldots)$ be the $v$-levelling and $N_k$ be a colorful level, consequently $\chi(N_k)\geq c_1+1=c/2$.
Apply Lemma \ref{lem: rooted PCP} and grow inside $N_k$ a strong rooted PCP $P=(G_1, P_1, H)$ of order 1 with leftovers at least $c_2=f(c_1)$. Then apply Lemma~\ref{lem: shadow borne bornee} and get an active lift $A=N(G_1)$ of $P$ inside $N_{k-1}$ with chromatic number at least $c_3=\varphi(c_2)$. By Lemma \ref{lem: rooted PCP}, we can obtain a strong rooted PCP $P'=(G'_1, P'_1, H')$ inside $A$ with leftovers at least $c_4=f(c_3-1)$, and by Lemma~\ref{lem: shadow borne bornee} we obtain an active lift $A'$ of $P'$ inside $N_{k-2}$ with chromatic number at least $c_5=\varphi(c_4)$. 
Because of the chromatic restriction in the definition of the PCP, one can color $G'_1$ with 4 colors. Moreover, $G'_1$ dominates $A'$ by definition. Thus there exists a stable set $S\subseteq P'$ such that $\chi(N(S)\cap A')\geq c_6=c_5/4$. 
Now $c_6\geq c'$ thus Lemma \ref{lem: sans C5} proves the existence of a 5-hole $C$ inside $N(S) \cap A'$. Let us give an overview of the situation: we have a 5-hole $C$ inside $N_{k-2}$, dominated by a stable set $S$ inside $N_{k-1}$, and every pair of vertices $t,t'$ of $S$ can be linked by a $tt'$-path $P_{down}$ with interior in $G_1\subseteq N_k$. Lemma \ref{lem: C5 chemin pair chemin impair} gives the existence of two vertices $t, t'\in S$ linked by both an odd path and an even path of length $\geq 3$ with interior in $C$. Closing one of these paths with $P_{down}$ gives an induced even hole of length $\geq 6$, a contradiction.
\end{proof}

\section*{Concluding remark}

Observe that no optimization was made on the constants $c'$ and $c$ from Lemma~\ref{lem: sans C5} and Theorem~\ref{th: C chi borne}. The proof gives the following upper bounds:
\begin{itemize}
\item $\chi(G) \leq 435122$ for every $G\in \CCinq$, and
\item $\chi(G)\leq 12\cdot 10^{18}$ for every $G\in \C$.
\end{itemize}

\section*{Acknowledgment}

The author would like to sincerely thank St\'ephan \textsc{Thomass\'e} for bringing this problem to her knowledge and for useful discussions about Trinity/Parity Changing Paths.

\bibliographystyle{plain}
\bibliography{biblio}

\end{document}